\documentclass[oneside,american,english]{amsart}
\usepackage[T1]{fontenc}
\usepackage[utf8]{inputenc}
\usepackage{mathrsfs}
\usepackage{amsthm}
\usepackage{amstext}
\usepackage{amssymb}

\makeatletter
\numberwithin{equation}{section}
\numberwithin{figure}{section}
\theoremstyle{plain}
\newtheorem{thm}{\protect\theoremname}
  \theoremstyle{plain}
  \newtheorem{lem}[thm]{\protect\lemmaname}
  \theoremstyle{plain}
  \newtheorem{cor}[thm]{\protect\corollaryname}
  \theoremstyle{remark}
  \newtheorem*{rem*}{\protect\remarkname}
  \theoremstyle{plain}
  \newtheorem{prop}[thm]{\protect\propositionname}

\usepackage{algorithmic}

\makeatother

\usepackage{babel}
  \addto\captionsamerican{\renewcommand{\corollaryname}{Corollary}}
  \addto\captionsamerican{\renewcommand{\lemmaname}{Lemma}}
  \addto\captionsamerican{\renewcommand{\propositionname}{Proposition}}
  \addto\captionsamerican{\renewcommand{\remarkname}{Remark}}
  \addto\captionsamerican{\renewcommand{\theoremname}{Theorem}}
  \addto\captionsenglish{\renewcommand{\corollaryname}{Corollary}}
  \addto\captionsenglish{\renewcommand{\lemmaname}{Lemma}}
  \addto\captionsenglish{\renewcommand{\propositionname}{Proposition}}
  \addto\captionsenglish{\renewcommand{\remarkname}{Remark}}
  \addto\captionsenglish{\renewcommand{\theoremname}{Theorem}}
  \providecommand{\corollaryname}{Corollary}
  \providecommand{\lemmaname}{Lemma}
  \providecommand{\propositionname}{Proposition}
  \providecommand{\remarkname}{Remark}
\providecommand{\theoremname}{Theorem}

\begin{document}

\title{Autocorrelations of Binary Sequences and Run Structure }

\author{Jürgen Willms}

\email{willms.juergen@fh-swf.de}

\address{Institut für Computer Science, Vision and Computational Intelligence,
Fachhochschule Südwestfalen, D-59872 Meschede, Germany}
\begin{abstract}
We \foreignlanguage{american}{analyze} the connection between the
autocorrelation of a binary sequence and its run structure given by
the run length encoding. We show that both the periodic and the aperiodic
autocorrelation of a binary sequence can be formulated in terms of
the run structure. The run structure is given by the consecutive runs
of the sequence. Let $C=(C_{0},C_{1},\cdots,C_{n})$ denote the\emph{
}autocorrelation vector of a binary sequence and $\triangle$ the
difference operator. We prove that the $k$th component of $\triangle^{2}(C)$
can be directly calculated by using the consecutive runs of total
length $k$. In particular this shows that the $k$th autocorrelation
is already determined by all consecutive runs of total length $l<k$.
In the aperiodic case we show how the run vector $R$ can be efficiently
calculated and give a characterization of skew-symmetric sequences
in terms of their run length encoding.
\end{abstract}


\date{14.4.2013}

\maketitle

\section{Introduction}

Let $n$ be a positive integer and and let $a=(a_{1},a_{2},\cdots,a_{n})$
be a (finite) sequence of real numbers. The \emph{length} $n$ of
the sequence $a$ will be denoted by $|a|$. $a$ is called a \emph{binary}
sequence if $a_{i}\in\{-1,1\}$ for all $i=1,\cdots,n$. In the following
we analyze in detail the connection between the autocorrelation of
a binary sequence and its run structure given by the run length encoding.
Binary sequences with suitable autocorrelation properties play an
important part in a wide range of different engineering applications.
For example they are used in signal processing in order to detect
signals in a noisy background. The autocorrelation measures the similarity
between the original sequence and its translate. In many applications
it is of interest to collectively minimize the absolute values of
the off-peak autocorrelations. For a survey on this topic we refer
to \cite{borwein2008merit,jedwab2005surveyspringerlink,jungnickel1999perfect}.
Depending on the type of application there are two types of autocorrelations
commonly used: the aperiodic and the periodic autocorrelation.

For $k=0,1,\cdots,n-1$ the $k$th \emph{aperiodic autocorrelation}
is given by 
\begin{equation}
C_{k}(a):=\sum_{i=1}^{n-k}a_{i}a_{i+k}.\label{eq:ck}
\end{equation}
In the periodic case put $a_{n+i}:=a_{i}$ for $i\geq1$; for $k=0,1,\cdots,n-1$
the $k$th \emph{periodic autocorrelation }is then defined by

\begin{equation}
\tilde{C}_{k}(a):=\sum_{i=1}^{n}a_{i}a_{i+k}.\label{eq:pck}
\end{equation}
In the following we will additionally put $C_{n}(a):=0$ and $\tilde{C}_{n}(a):=n$.
Note that for a binary sequence the peak autocorrelation equals the
length of the sequence: $C_{0}(a)=\tilde{C}_{0}(a)=n$. Note further
that for $k=1,\cdots,n$ we have $\tilde{C}_{k}(a)=\tilde{C}_{n-k}(a)$
and 
\begin{equation}
\tilde{C}_{k}(a)=C_{k}(a)+C_{n-k}(a).\label{eq:Ckperiodic}
\end{equation}

In the following we analyze in detail the connection between the autocorrelations
of a binary sequence and its run structure. A run is defined as a
substring of maximal length where all elements have the same value.
Runs as well as autocorrelation values were used in \cite{golomb1967shift}
in order to measure apparent randomness in a binary sequence (according
to \cite{golomb2005signal} these randomness postulates first appeared
in \cite{golomb1955sequences}) .

We prove that both the periodic and the aperiodic autocorrelation
of a binary sequence can be formulated in terms of the run structure.
As we will see the run structure is given by the consecutive runs
of $a.$ Consecutive runs of $a$ with total length $k$ determine
$R_{k}$, the $k$th element of $R$; we call $R$ the run vector.
If $C:=(C_{0}(a),C_{1}(a),\cdots,C_{n}(a))$ denotes the\emph{ }autocorrelation
vector of $a$ and $\triangle$ the forward difference operator, then
we will prove that $\triangle^{2}(C)=-2R.$ In particular this shows
that the consecutive runs of total length up to $k-1$ determine the
$k$th autocorrelation. 

For the periodic case this was established in \cite{cai2009autocorrelation}.
The objective of this paper is twofold: to simplify the proof in \cite{cai2009autocorrelation}
and to find a similar relationship for the aperiodic autocorrelation.
We start with the latter; in the first part of this paper the aperiodic
case is considered and it is shown how the aperiodic autocorrelation
of a binary sequence can be formulated in terms of the run structure.
In the aperiodic case we show how the run vector can be efficiently
calculated based on a simple algorithm and we derive a further practical
formula for calculating the elements of the run vector. Furthermore,
we give a characterization of skew-symmetric sequences in terms of
their run length encoding. Finally, we consider the periodic case.
The presented proof of the aperiodic case can also be applied to the
periodic case with minor modifications, resulting in a new simplified
and more direct proof. This paper, however, follows a different and
even shorter approach: the results for the periodic case are directly
derived from the aperiodic case by using (\ref{eq:Ckperiodic}). 

After the completion of this paper the author became aware that a
different run correlation technique for the aperiodic case was published
in \cite{polge673421}. There a tabular arrangement for the correlation
calculation was developed by using a set of sequential relations which
relate the aperiodic autocorrelation to the run structure. This is
essentially the result of Theorem \ref{thm:main_aperiodic} which
we prove in the first part of this paper. In \cite{polge1986} these
results were used for a search strategy in order to construct binary
sequences with specified aperiodic autocorrelation values by explicitly
eliminating large subsets of binary sequences and thus reducing the
search space.

\section{Preliminaries }

In the following $a$ will always be a fixed binary sequence of length
$n$. Furthermore, in the aperiodic case we will always put $a_{0}:=0$
and $a_{n+1}:=0$ in order to circumvent boundary problems. 

A \emph{substring} of $a$ always represents a non-empty contiguous
part $(a_{i},a_{i+1},\cdots,a_{j-1})$ of $a$ with $1\leq i<j\leq n+1$;
quite similar to the concept of an half-open interval it will be denoted
by $a(i,j)$. In the following we will also always distinguish between
the substrings $a(i,j)$ and $a(i',j')$ whenever $(i,j)\neq(i',j')$.
Hence $a$ has $\frac{n\cdot(n+1)}{2}$ different substrings. Note
that by our definition a substring is always non-empty and that the
length $|a(i,j)|$ of the substring $a(i,j)$ is given by $|a(i,j)|=j-i$.
If $1<i$ and $j<n+1$, then $a(i,j)$ is called an \emph{inner} substring,
otherwise\emph{ }$a(i,j)$ is called an \emph{outer} substring.

\subsection{Runs, Run Blocks and Run Length Encoding}

A \emph{run} of $a$ is a substring of $a$ with maximal length where
all elements have the same value. Thus if $a(i,j)$ is a run of $a$,
then $a_{i-1}\neq a_{i}=a_{i+1}=\cdots=a_{j-1}\neq a_{j}$. In the
following $\gamma$ will always denote the total number of runs of
the sequence $a$. Consecutive runs form what we will call a \emph{run
block}. Thus for $1\leq i<j\leq n+1$ a substring $a(i,j)$ is a run
block\emph{ of $a$,} if and only if $a_{i-1}\neq a_{i}$ and $a_{j-1}\neq a_{j}$.
Let for example $a$ be the binary sequence of length 13 given by
$a=(+++++++---+++)$; here and in the following the symbol $'+'$
stands for 1 and the symbol $'-'$ for -1. In this case we have $\gamma=3$
since $a$ has three runs, namely $a(1,8)$, $a(8,11)$ and $a(11,14)$.
Furthermore, $a$ has a total of six run blocks: there are the three
run blocks $a(1,8)$, $a(8,11)$ and $a(11,14)$ consisting of just
a single run, there are the two run blocks $a(1,11)$ and $a(8,14)$
consisting of two consecutive runs and there is the run block $a(1,14)$
which is $a$ itself consisting of three consecutive runs.

For an inner run block $a(i,j)$ (i.e.\ a run block which is an inner
substring such as, for example $a(8,11)$ of the previous example)
we have $a_{i-1}=-a_{i}$ and $a_{j-1}=-a_{j}$. Note since $a_{0}=a_{n+1}=0$
the sequence $a$ itself is represented by the run block $a(1,n+1)$
and that each run block of $a$ can be uniquely divided in runs of
$a$. In particular, the sequence $a$ can be uniquely divided in
$\gamma$ runs $a(i_{k},i_{k+1})$ with $1=i_{1}<i_{2}<\cdots<i_{\gamma+1}=n+1$;
$a(i_{k},i_{k+1})$ is then called the \emph{$k$th run} of $a$.
If we put $r_{k}:=|a(i_{k},i_{k+1})|=i_{k+1}-i_{k}$, then the sequence
$r=(r_{1},r_{2},\cdots,r_{\gamma})$ is called the \emph{run length
encoding} of $a$. Note that beside $a$ only the binary sequence
$(-a_{1},-a_{2},\cdots,-a_{n})$ has the same run length encoding
as $a$.

Let for example $a$ be the binary sequence of length 13 given by
$a=(++++++-------)$, then $\gamma=2$ and the run length encoding
$r$ of $a$ is given by $r=(6,7)$. In this case $a$ has two runs
and three blocks. If $a=(+++++++---+++)$, then the run length encoding
of $a$ is given by $r=(7,3,3)$ and as already noted $a$ has three
runs and six run blocks. The binary sequence of $a=(+++------+++---)$
has four runs and its run length encoding is given by $r=(3,6,3,3)$. 

As we will see later the weight of a run block will be used in order
to calculate the autocorrelation. The \emph{weight} $w$ of a substring
$a(i,j)$ is defined by 
\begin{equation}
w(a(i,j)):=\begin{cases}
2a_{i}\cdot a_{j-1}\; & \mbox{if }a(i,j)\mbox{ is an inner run block}\\
a_{i}\cdot a_{j-1} & \mbox{if }a(i,j)\mbox{ is an outer run block}\\
0 & \mbox{otherwise}.
\end{cases}\label{eq:weight}
\end{equation}

Note that for a substring $b$ we have $w(b)=0$ unless $b$ is a
run block. In this case $|w(b)|=2$ if $b$ is an inner and $|w(b)|=1$
if $b$ is an outer run block. We remark further, that if a run block
$b=a(i,j)$ consists of $m$ consecutive runs, then $a_{i}\cdot a_{j-1}=-(-1)^{m}$
and thus $w(b)>0$ if $m$ is odd and $w(b)<0$ if $m$ is even.

\subsection{The Aperiodic Run Structure}

In the following $r$ will always denote the run length encoding of
$a$ with $r_{k}=i_{k+1}-i_{k}$ and $\gamma$ as defined above. Let
us define the \emph{aperiodic} \emph{run structure} $\mathsf{\mathcal{R}}$
of $a$ as the set of all substrings of $r$. 

Next we want to show that there is a one-to-one correspondence between
the run blocks of $a$ and the substrings of the run length encoding
$r$. If $b$ is a run block of $a$ of length $k$, then we have
$b=a(i_{p},i_{q})$ for some $1\leq p<q\leq\gamma+1$ and the run
block $b=a(i_{p},i_{q})$ corresponds to the substring $r(p,q)$.
The run length encoding gives us therefore a mapping $\Phi$ from
the set $\mathbb{\mathsf{\mathscr{\mathcal{B}}}}$ of all run blocks
of $a$ to the aperiodic run structure $\mathsf{\mathcal{R}}$ of
$a$ defined by $\Phi(a(i_{p},i_{q}))=r(p,q)$. Note that $\Phi$
is bijective and that $\Phi$ maps an inner run block of $a$ to an
inner substring of $r$. Furthermore for $k=1,\cdots,n$ let $\mathbb{\mathsf{\mathscr{\mathcal{B}}}}_{k}$
denote the set of all run blocks of $a$ with length $k$ and let
$\mathsf{\mathcal{R}}_{k}$ denote the set of all substrings $r(p,q)$
of $r$ whose sum $\sum_{j=p}^{q-1}r_{j}$ is equal to $k$; obviously
$\mathbb{\mathsf{\mathscr{\mathcal{B}}}}_{k}\subseteq\mathfrak{\mathcal{B}}$
and $\mathsf{\mathcal{R}}_{k}\subseteq\mathcal{R}.$ If $b=a(i_{p},i_{q})$
is a run block of $a$ with length $k$, then we have 
\begin{equation}
i_{q}-i_{p}=\sum_{j=p}^{q-1}r_{j}=k\label{eq:iq}
\end{equation}
and hence 
\begin{equation}
\Phi(\mathbb{\mathsf{\mathscr{\mathcal{B}}}}_{k})=\mathsf{\mathcal{\mathcal{R}}}_{k}.\label{eq:map1}
\end{equation}
 Thus each element $u\in\mathsf{\mathcal{\mathcal{R}}}_{k}$ (i.e.\ each
substring $u$ of $r$ whose sum is equal to $k$) corresponds uniquely
to a run block of length $k$ consisting of $|u|$ consecutive runs,
and vice versa. 

For a substring $u$ of $r$ let

\begin{equation}
\alpha(u):=\begin{cases}
2\; & \mbox{if }u\mbox{ is an inner substring of \ensuremath{r}}\\
1\; & \mbox{otherwise}.
\end{cases}\label{eq:alpha-def}
\end{equation}

If the run block $a(i_{p},i_{q})$ consists of $m$ consecutive runs,
then $a_{i_{p}}\cdot a_{i_{q}-1}=-(-1)^{m}$ as already noted. Since
$\Phi(a(i_{p},i_{q}))=r(p,q)$ we have $m=|r(p,q)|$ and it follows
from (\ref{eq:weight}) that $w(a(i_{p},i_{q}))=-\alpha(r(p,q))\cdot(-1)^{|r(p,q)|}$.
Since the mapping $\Phi$ is bijective, (\ref{eq:map1}) together
with the last remark of the previous subsection gives us that
\begin{equation}
\sum_{u\in\mathcal{\mathit{\mathcal{R}}}_{k}}\alpha(u)\cdot(-1)^{|u|}=-\sum_{b\in\mathcal{B}_{k}}w(b).\label{eq:wbls}
\end{equation}

\section{The Main Result for the Aperiodic Case}

In this section we want to analyze the connection between the aperiodic
autocorrelations $C_{k}(a)$ of the binary sequence $a$ and its run
structure given by the run length encoding $r=(r_{1},r_{2},\cdots,r_{\gamma})$.
Note that $\sum_{j=1}^{\gamma}r_{j}=|a|=n=C_{0}(a)$. Moreover, we
have
\begin{equation}
C_{1}(a)=n+1-2\gamma\label{eq:C1}
\end{equation}
since for each $j=1,2,\cdots,\gamma-1$ the $j$th run of $a$ contributes
$r_{j}-2$ to the sum $C_{1}(a)$ whereas the last run contributes
$r_{\gamma}-1$ to the sum $C_{1}(a)$. Hence we have $C_{1}(a)=\sum_{j=1}^{\gamma-1}(r_{j}-2)+(r_{\gamma}-1)=1+\sum_{j=1}^{\gamma}(r_{j}-2)=1+n-2\gamma$.

For $k=1,2,\cdots,n-1$ put 
\begin{equation}
R_{k}:=\sum_{u\in\mathcal{\mathit{\mathcal{R}}}_{k}}\alpha(u)\cdot(-1)^{|u|}\label{eq:RkDef}
\end{equation}
and $R(a):=(R_{1},R_{2},\cdots,R_{n-1})$; we call $R(a)$ the \emph{run
vector} of $a$. Thus by (\ref{eq:wbls}) we have 
\begin{equation}
R_{k}=-\sum_{b\in\mathcal{B}_{k}}w(b).\label{eq:RkBlock}
\end{equation}
Therefore, in order to calculate $R_{k}$ all consecutive runs of
$a$ with a total length of $k$ (i.e.\ all run blocks of $a$ with
length $k$) have to be considered. Each run block $b$ of length
$k$ contributes the (negative) weight $-w(b)$ as defined in (\ref{eq:weight})
to the sum in (\ref{eq:RkBlock}). As already noted $|w(b)|$ equals
1 if $b$ is an outer run block, $|w(b)|$ equals 2 if $b$ is an
inner run block and the sign of $w(b)$ depends only on the number
of runs: we have $w(b)>0$ if the run block $b$ consists of an odd
number of consecutive runs and $w(b)<0$ if $b$ consists of an even
number of consecutive runs.

Let us for example compute the run vector $R(a)$ for the three sequences
of the previous example. If $r=(6,7)$, then apart from $a$ itself
there are only two run blocks; they both consist of a single run,
their length is 6 resp.\ 7 and both are outer run blocks. Thus we
have $\mathsf{\mathcal{R}}_{k}=\textrm{Ø}$ for $k\neq6,\,7$ and
$\mathsf{\mathcal{R}}_{6}=\left\{ r(1,2)\right\} $, $\mathsf{\mathcal{R}}_{7}=\left\{ r(2,3)\right\} $;
hence by (\ref{eq:RkDef}) $R(a)=(0,0,0,0,0,-1,-1,0,0,0,0,0)$.

For $r=(7,3,3)$ there are three run blocks consisting of a single
run; their length is 7, 3 and 3, their weight -1, -2 and -1 and they
correspond to the substrings $r(1,2)$, $r(2,3)$ and $r(3,4)$. In
addition there are two (outer) run blocks consisting of two consecutive
runs. Their length is 10 resp.\ 6, both have weight 1 and they correspond
to the substrings $r(1,3)$ and $r(2,4)$. Altogether, we have $\mathsf{\mathcal{R}}_{3}=\left\{ r(2,3),r(3,4)\right\} $,
$\mathsf{\mathcal{R}}_{6}=\left\{ r(2,4)\right\} $, $\mathsf{\mathcal{R}}_{7}=\left\{ r(1,2)\right\} $,
$\mathsf{\mathcal{R}}_{10}=\left\{ r(1,3)\right\} $ and $\mathsf{\mathcal{R}}_{k}=\textrm{Ø}$
for the remaining cases $k=1,2,4,5,8,9,11,12$; this shows that $R(a)=(0,0,-3,0,0,1,-1,0,0,1,0,0)$.

Finally, we consider the example $r=(3,6,3,3)$. In order to calculate
for instance $R_{6}$ we have to consider all consecutive runs of
$a$ which have a total length (i.e.\ run block length) of 6. There
are exactly two run blocks of $a$ with length 6. The first one consists
of a single run, has weight -2 and corresponds to the inner substrings
$r(2,3)$; the second one consists of two consecutive runs, has weight
1 and corresponds to the outer substring $r(3,5)$. This shows that
$\mathsf{\mathcal{R}}_{6}=\left\{ r(2,3),r(3,5)\right\} $ and $R_{6}=-1$.
Similarly, we have $\mathsf{\mathcal{R}}_{3}=\left\{ r(1,2),r(3,4),r(4,5)\right\} $,
$\mathsf{\mathcal{R}}_{9}=\left\{ r(1,3),r(2,4)\right\} $, $\mathsf{\mathcal{R}}_{12}=\left\{ r(1,4),r(2,5)\right\} $
and $\mathsf{\mathcal{R}}_{k}=\textrm{Ø}$ if $k$ is not a multiple
of 3. Hence it follows that $R(a)=(0,0,-4,0,0,-1,0,0,3,0,0,-2,0,0)$. 

The next theorem (cf.\ \cite{polge673421}) is the main result for
the aperiodic case. It shows that for a binary sequence $a$ the aperiodic
autocorrelations and and the run vector $R(a)$ are closely related.
\begin{thm}
\label{thm:main_aperiodic}$\mbox{Let }k=1,\cdots,n-1$; then
\[
C_{k+1}(a)-2C_{k}(a)+C_{k-1}(a)=-2R_{k}.
\]
\end{thm}
\begin{proof}
Let $\delta:=(\delta_{1},\delta_{2},\cdots,\delta_{n+1})$ be the
sequence defined by $\delta_{i}:=a_{i}-a_{i-1}$ for all $i=1,2,\cdots,n+1.$
Since $a(i,j)$ is a run block of\emph{ $a$ }if and only if $a_{i-1}\neq a_{i}$
and $a_{j-1}\neq a_{j}$ it follows that
\begin{equation}
a(i,j)\mbox{ is a block of }a\;\;\Leftrightarrow\;\;\delta_{i}\cdot\delta_{j}\neq0\label{eq:block}
\end{equation}

Now let $1\leq i<j\leq n+1$ with $j-i<n.$ If $a(i,j)$ is a run
block, then $a_{i}a_{j}+a_{i-1}a_{j-1}=-a_{i}a_{j-1}-a_{i-1}a_{j}$;
with (\ref{eq:block}) this shows that 
\begin{eqnarray*}
\delta_{i}\delta_{j} & = & (a_{i}-a_{i-1})(a_{j}-a_{j-1})\\
 & = & a_{i}a_{j}-a_{i}a_{j-1}-a_{i-1}a_{j}+a_{i-1}a_{j-1}\\
 & = & -2(a_{i}a_{j-1}+a_{i-1}a_{j})\\
 & = & -2w(a(i,j))
\end{eqnarray*}

Hence by (\ref{eq:wbls}) 
\begin{eqnarray*}
2R_{k} & = & 2\sum_{u\in\mathcal{\mathit{\mathcal{R}}}_{k}}\alpha(u)\cdot(-1)^{|u|}=-2\sum_{b\in\mathcal{B}_{k}}w(b)\\
 & = & -2\sum_{i=1}^{n+1-k}w(a(i,i+k))\\
 & = & \sum_{i=1}^{n+1-k}\delta_{i}\delta_{i+k}=C_{k}(\delta).
\end{eqnarray*}
Now Theorem \ref{thm:main_aperiodic} follows directly from the following
lemma.\end{proof}
\begin{lem}
\label{lem:delta-Ck}Let $\delta=(\delta_{1},\delta_{2},\cdots,\delta_{n+1})$
be the sequence defined by $\delta_{i}:=a_{i}-a_{i-1}$ for all \textup{$i=1,2,\cdots,n+1$;}
then \textup{$\mbox{for }k=1,2,\cdots,n-1$} 
\[
C_{k+1}(a)-2C_{k}(a)+C_{k-1}(a)=-C_{k}(\delta).
\]
 \end{lem}
\begin{proof}
Let $1\leq k\leq n-1$; then 
\[
\begin{split}C_{k}(\delta)\\
= & \sum_{i=1}^{n+1-k}\delta_{i}\delta_{i+k}=\sum_{i=1}^{n+1-k}(a_{i}-a_{i-1})(a_{i+k}-a_{i+k-1})\\
= & \sum_{i=1}^{n+1-k}a_{i}a_{i+k}-a_{i}a_{i+k-1}-a_{i-1}a_{i+k}+a_{i-1}a_{i+k-1}\\
= & (C_{k}(a)+a_{n+1-k}a_{n+1})-C_{k-1}(a)\\
 & -(a_{0}a_{k+1}+C_{k+1}(a)+a_{n-k}a_{n+1})\\
 & +(a_{0}a_{k}+C_{k}(a))\\
= & -C_{k+1}(a)+2C_{k}(a)-C_{k-1}(a).
\end{split}
\]

\end{proof}
We have $C_{0}(a)=n$ and $C_{1}(a)=n+1-2\gamma$ by (\ref{eq:C1});
Theorem \ref{thm:main_aperiodic} gives us that $C_{2}(a)=2C_{1}(a)-n-2R_{1}$
and thus $C_{2}(a)=n+2-4\gamma-2R_{1}$. Furthermore Theorem \ref{thm:main_aperiodic}
shows that $C_{k+1}(a)=2C_{k}(a)-C_{k-1}(a)-2R_{k}\;\;\mbox{for }k=1,2,\cdots,n-1$. 

Let us denote by $C(a):=(C_{0}(a),C_{1}(a),\cdots,C_{n}(a))$ the
\emph{aperiodic autocorrelation vector} of $a$. We can now use Theorem
\ref{thm:main_aperiodic} in order to compute $C(a)$ for the three
sequences of the previous examples. As we have seen, if $r=(6,7)$
then we have $R(a)=(0,0,0,0,0,-1,-1,0,0,0,0,0)$. It follows that
$C_{0}(a)=n=13$, $C_{1}(a)=n+1-2\gamma=13+1-2\cdot2=10$, $C_{2}(a)=2C_{1}(a)-n-2R_{1}=20-13=7$
, $C_{3}(a)=2C_{2}(a)-C_{1}(a)-2R_{2}=4$ and so on, which gives us
$C(a)=(13,10,7,4,1,-2,-5,-6,-5,-4,-3,-2,-1,0)$. Furthermore, by a
simple calculation we get $C(a)=(13,8,3,-2,-1,0,1,0,1,2,3,2,1,0)$
if $r=(7,3,3)$ and $C(a)=(15,8,1,-6,-5,-4,-3,0,3,6,3,0,-3,-2,-1,0)$
if $r=(3,6,3,3)$. 

Theorem \ref{thm:main_aperiodic} can be rephrased by using the difference
operator. The (\emph{forward}) \emph{difference operator} $\triangle(b)$
of a sequence $b=(b_{1},b_{2},\cdots,b_{m})$ with $m\geq2$ is given
by $\triangle(b):=(b_{2}-b_{1},b_{3}-b_{2},\cdots,b_{m}-b_{m-1})$.
For $m\geq3$ we have $\triangle^{2}(b):=\triangle(\triangle(b))=(b_{3}-2b_{2}+b_{1},b_{4}-2b_{3}+b_{2},\cdots,b_{m}-2b_{m-1}+b_{m-2})$.
\begin{cor}
\label{cor:-delta}$\triangle^{2}(C(a))=-2R(a)\;$ for $n\geq3$.\end{cor}
\begin{proof}
This is just a reformulation of Theorem \ref{thm:main_aperiodic}.\end{proof}
\begin{rem*}
If we put $\overline{a}:=(0,a_{1},a_{2},\cdots,a_{n},0)$, then Lemma
\ref{lem:delta-Ck} shows that for $n\geq3$
\[
\triangle^{2}(C(a))=-(C_{1}(\triangle(\overline{a})),C_{2}(\triangle(\overline{a})),\cdots,C_{n-1}(\triangle(\overline{a}))).
\]
 A more explicit relationship between the aperiodic autocorrelation
and the run vector gives the next result. Note that by (\ref{eq:C1})
$C_{1}(a)-C_{0}(a)=1-2\gamma$. \end{rem*}
\begin{cor}
\label{cor:-Ck}$\mbox{Let }k=0,1,\cdots,n$; then 
\[
C_{k}(a)=n+(1-2\gamma)k-2\sum_{j=1}^{k-1}(k-j)R_{j}.
\]
\end{cor}
\begin{proof}
Since $C_{0}(a)=n$ and by (\ref{eq:C1}) $C_{1}(a)=n+1-2\gamma$
the statement is true for $k=0,1$. For $k\geq2$ it follows directly
from $\triangle^{2}(C(a))=-2R(a)$ and the next Lemma.\end{proof}
\begin{lem}
\label{lem:delta}For $m\geq3$ let $b=(b_{1},b_{2},\cdots,b_{m})$
be a sequence of real numbers and let $\epsilon_{j}$ denote the $j$th
component of $\triangle^{2}(b)$, i.e.\textup{\ $\triangle^{2}(b)=(\epsilon_{1},\epsilon_{2},\cdots,\epsilon_{m-2})$;}
then for $k=0,1,\cdots,m-1$ 
\[
b_{k+1}=b_{1}+(b_{2}-b_{1})k+\sum_{j=1}^{k-1}(k-j)\epsilon_{j}.
\]
\end{lem}
\begin{proof}
This can be easily proved by induction on $k$.
\end{proof}
Next we want to show that based on the symmetry of the operator $\triangle^{2}$
a complementary formulation of Corollary \ref{cor:-Ck} can be derived.
If $b$ and $\triangle^{2}(b)=(\epsilon_{1},\epsilon_{2},\cdots,\epsilon_{m-2})$
are defined as in Lemma \ref{lem:delta}, then we also have for $k=0,1,\cdots,m-1$
\begin{equation}
b_{k+1}=b_{m}+(b_{m-1}-b_{m})(m-k-1)+\sum_{j=k+1}^{m-2}(j-k)\epsilon_{j}.\label{eq:bkplus1}
\end{equation}
The proof is similar to the proof of Lemma \ref{lem:delta}. By applying
this to $\triangle^{2}(C(a))=-2R(a)$ and noting that $C_{n-1}(a)=(-1)^{\gamma+1}$
and $C_{n}(a)=0$ it follows as in the proof of Corollary \ref{cor:-Ck}
that for $k=0,1,\cdots,n$ 
\begin{equation}
C_{k}(a)=(-1)^{\gamma+1}(n-k)-2\sum_{j=k+1}^{n-1}(j-k)R_{j}.\label{eq:Ck_corresponding}
\end{equation}

Similarly, from $\epsilon_{k}=b_{k+2}-2b_{k+1}+b_{k}$ it easily follows
that $\sum_{j=1}^{m-2}\epsilon_{j}=b_{1}-b_{2}+b_{m}-b_{m-1}$. Applying
this to $\triangle^{2}(C(a))=-2R(a)$ gives us that $2\sum_{k=1}^{n-1}R_{k}=C_{0}(a)-C_{1}(a)+C_{n}(a)-C_{n-1}(a)=1-2\gamma+(-1)^{\gamma+1}$
and hence 

\begin{equation}
\sum_{k=1}^{n-1}R_{k}=\begin{cases}
-\gamma\; & \mbox{if }\gamma\mbox{ even}\\
1-\gamma\; & \mbox{\mbox{if }\ensuremath{\gamma}\mbox{ odd}.}
\end{cases}\label{eq:sum_rk}
\end{equation}

\section{Applying the Results}

In this section we will first present an example where the run vector
$R(a)$ is used in order to establish a relationship between the correlations
of certain related sequences. Next we show how the run vector $R(a)$
can be efficiently calculated based on a simple algorithm. After that
we develop an alternative formula which shows how $R_{k}$ can be
expressed in terms of the sums $r_{1}+r_{2}+\cdots+r_{j}$. The rest
of this section discusses skew-symmetric sequences and their run length
encoding. For a skew-symmetric sequence $a$ we have $C_{k}(a)=0$
whenever $k$ is odd; if $k>0$ is even, then we will show that $C_{k}(a)=R_{k}(a)$.
Moreover, we give a characterization of skew-symmetric sequences in
terms of their run length encoding.

\subsection{A First Example}

For $m>1$ we consider the sequence $b=(b_{1},b_{2},\cdots,b_{nm})$
obtained by repeating each element of $a$ exactly $m$ times so that
the run length encoding $r(b)$ of $b$ is given by $(mr_{1},mr_{2},\cdots,mr_{\gamma})$.
For example $m=2$ gives us $b=(a_{1},a_{1},a_{2},a_{2},\cdots,a_{n},a_{n})$.
We will show that for $0\leq k<n$ and $0\leq s<m$ the aperiodic
autocorrelations of $b$ are given by 
\begin{equation}
C_{km+s}(b)=(m-s)C_{k}(a)+sC_{k+1}(a)\label{eq:Cb_example}
\end{equation}
where we have put as before $C_{n}(a):=0$.

For $0\leq k<n$ and $0\leq s<m$ put $D_{km+s}:=(m-s)C_{k}(a)+sC_{k+1}(a)$,
$D_{nm}=0$ and $D:=(D_{0},D_{1},\cdots,D_{nm}).$ In order to prove
(\ref{eq:Cb_example}) using Lemma \ref{lem:delta} and Corollary
\ref{cor:-delta} it is sufficient to show that
\[
C_{0}(b)=D_{0}\mbox{ , }C_{1}(b)=D_{1}\mbox{ and }\triangle^{2}(D)=-2R(b).
\]

We have $C_{0}(b)=mn=mC_{0}(a)=D_{0}$ and by (\ref{eq:C1}) $C_{1}(b)=mn+1-2\gamma=(m-1)n+n+1-2\gamma=(m-1)C_{0}(a)+C{}_{1}(a)=D_{1}$.
Since $r(b)=(mr_{1},mr_{2},\cdots,mr_{\gamma})$ it is not difficult
to see that for $1\leq j<nm$ 
\[
R_{j}(b)=\begin{cases}
R_{\frac{j}{m}}(a) & \;\mbox{if }j\equiv0\; mod\; m\\
0 & \;\mbox{otherwise.}
\end{cases}
\]

Now let $\triangle^{2}(D)=(\epsilon_{1},\epsilon_{2},\cdots,\epsilon_{nm-1}).$
For $0\leq k<n$ we have $D_{km+s+1}-D_{km+s}=C_{k+1}(a)-C_{k}(a)$
for $0\leq s<m-1$ and also for $s=m-1$. Since $\triangle^{2}(D):=\triangle(\triangle(D))$
we have $\epsilon_{j}=0$ unless $j$ is a multiple of $m$; in this
case we have $\epsilon_{i\cdot m}=C_{i+1}(a)-2C_{i}(a)+C_{i-1}(a)=-2R_{i}(a)=-2R_{i\cdot m}(b)$.
Hence $\triangle^{2}(D)=-2R(b)$.

\subsection{Calculation of the Run Vector}

The results of the previous section show how the run structure and
the aperiodic autocorrelations are related. The main result which
says that $\triangle^{2}(C(a))=-2R(a)$ can be used in order to calculate
the autocorrelation vector $C(a)$. However, the presented form, in
particular (\ref{eq:RkDef}), is not very well suited for practical
purposes. However, $R(a)$ can be efficiently computed based on the
following simple algorithm, which can be easily derived from (\ref{eq:RkDef}).
As a precondition we will assume that the run length encoding $r$
and its length $\gamma$ are known and that each component of the
array $R$ is initialized to zero. After the final step of the algorithm
the array $R$ contains the computed run vector $R(a).$ The algorithm
computes $R(a)$ in two steps. The first step considers only the outer
substrings of the run length encoding $r$ of $a$ and uses the fact
that $|r(1,j)|=\gamma-|r(j,\gamma+1)|$: 

\smallskip{}

\begin{algorithmic} 

\STATE $\hat{\gamma}\leftarrow(-1)^{\gamma}$ 

\STATE$\alpha\leftarrow-1$

\STATE $s\leftarrow0$

\FOR{$j=1$ \TO $\gamma-1$}

\STATE $s\leftarrow s+r_{j}$

\STATE $R_{s}\leftarrow R_{s}+\alpha$

\STATE $R_{n-s}\leftarrow R_{n-s}+\hat{\gamma}\cdot\alpha$

\STATE $\alpha\leftarrow-\alpha$

\ENDFOR

\smallskip{}

The second and final step of the algorithm takes into account all
inner substrings of $r$:

\smallskip{}

\FOR{$i=2$ \TO $\gamma-1$}

\STATE$\alpha\leftarrow-2$

\STATE $s\leftarrow0$

\FOR{$j=i$ \TO $\gamma-1$}

\STATE $s\leftarrow s+r_{j}$

\STATE $R_{s}\leftarrow R_{s}+\alpha$ 

\STATE $\alpha\leftarrow-\alpha$

\ENDFOR

\ENDFOR

\end{algorithmic} 

\smallskip{}

In the first step of the algorithm there are $\gamma-1$ iterations
and we have a total of $\frac{(\gamma-1)\cdot(\gamma-2)}{2}$ iterations
in the second step. In the second step by enrolling the inner loop
into two separate loops any multiplication can be avoided; the same
is true for the loop in the first step. Thus the computation of $R(a)$
requires $4(\gamma-1)+$\allowbreak$(\gamma-1)(\gamma-2)=(\gamma-1)(\gamma+2)$
additions and no multiplications. 

Let us compare this with the direct calculation of the autocorrelations
$C_{1}(a),C_{2}(a),$\allowbreak$\cdots,C_{n-1}(a)$ as defined in
(\ref{eq:ck}) which requires $\frac{n\cdot(n-1)}{2}$ multiplications
and $\frac{(n-1)\cdot(n-2)}{2}$ additions. Now, at this point let
us assume that $\gamma\approx\frac{n}{2}$; we will come back to this
assumption below. If $\gamma\approx\frac{n}{2}$ then the above algorithm
for the computations of $R(a)$ requires approximately $\frac{n^{2}}{4}$
additions and no multiplication. Even if we do not distinguish between
the cost of an addition and a multiplication, then the above algorithms
should be approximately four times faster than the direct calculation
using (\ref{eq:ck}); a very similar result can be found in \cite{polge673421}. 

At least for large $n$ this remains true, even if we consider the
additional cost in order to compute $C(a)$ and the run length encoding
$r$ of $a$: if $-2R(a)$ is already computed then we need less than
$2n$ additions in order to calculate the autocorrelation vector $C(a)$
using $\triangle^{2}(C(a))=-2R(a)$ and not more than $n+\gamma$
additions are necessary in order to calculate $r$.

Let us come back to the assumption that $\gamma\approx\frac{n}{2}$.
Note that this is true for most binary sequences if $n$ is large.
If, however, $\gamma>\frac{n+1}{2}$, then consider the binary sequence
$\bar{a}$ of length $n$ one gets, when inverting every second element
of $a$, i.e.\ $\bar{a}_{i}:=a_{i}$ if $i$ is odd and $\bar{a}_{i}:=-a_{i}$
if $i$ is even. Then $C_{k}(\bar{a})=C_{k}(a)$ if $k$ is even and
$C_{k}(\bar{a})=-C_{k}(a)$ if $k$ is odd. By induction it is not
difficult to show that $\gamma+\bar{\gamma}=n+1$ where $\bar{\gamma}$
denotes the length of the run length encoding of $\bar{a}.$ Thus
if $\gamma>\frac{n+1}{2}$, then we have $\bar{\gamma}\leq\frac{n}{2}$
and $\bar{a}$ has up to every second sign the same aperiodic autocorrelation
vector. So in the above algorithm $\bar{a}$ instead of $a$ can be
used in this case in order to compute the autocorrelation vector $C(a)$.

\subsection{A More Practical Formula for $R_{k}$}

In this subsection we develop a formula which shows how $R_{k}$ can
be expressed in terms of the sums $r_{1}+r_{2}+\cdots+r_{j}$; this
gives rise to a different algorithm for calculating $R_{k}$ which
is for example particularly useful in a branch-and-bound like exhaustive
search where only parts of the sequence $a$ are known. 

For $j=1,2,\cdots,\gamma$ let
\begin{equation}
s_{j}:=r_{1}+r_{2}+\cdots+r_{j}\label{eq:s-def}
\end{equation}

and 
\begin{equation}
t_{j}:=r_{\gamma}+r_{\gamma-1}+\cdots+r_{\gamma-j+1}.\label{eq:t-def}
\end{equation}
Note that then $1\leq s_{1}<s_{2}<\cdots<s_{\gamma}=n$, $1\leq t_{1}<t_{2}<\cdots<t_{\gamma}=n$
and that for $j=1,2,\cdots,\gamma-1$
\begin{equation}
s_{j}+t_{\gamma-j}=n.\label{eq:s+t}
\end{equation}
Furthermore, let 
\begin{equation}
S:=\{s_{1},s_{2},\cdots,s_{\gamma-1}\}\label{eq:S_def}
\end{equation}
and
\begin{equation}
T:=\{t_{1},t_{2},\cdots,t_{\gamma-1}\}.\label{eq:T_def}
\end{equation}
The functions $f_{S},\: f_{T}:\;\mathbb{Z}\rightarrow\{-1,0,1\}$
defined by
\begin{equation}
f_{S}(k):=\begin{cases}
(-1)^{j}\; & \mbox{if }k\in S\mbox{ with \ensuremath{k=s_{j}}}\\
0\; & \mbox{otherwise}
\end{cases}\label{eq:f-def}
\end{equation}
\begin{equation}
f_{T}(k):=\begin{cases}
(-1)^{j}\; & \mbox{if }k\in T\mbox{ with \ensuremath{k=t_{j}}}\\
0\; & \mbox{otherwise}
\end{cases}\label{eq:g-def}
\end{equation}
 will play an important role in the following. Note that by (\ref{eq:s+t})
we have for all $k\mathbb{\in Z}$ 
\begin{equation}
f_{S}(k)=(-1)^{\gamma}f_{T}(n-k).\label{eq:f=00003Dg}
\end{equation}

The next theorem shows how $R_{k}$ can be expressed in terms of $s_{1},s_{2},\cdots,s_{\gamma-1}$.
\begin{thm}
\label{thm:Rk_f}\textup{Let $k=1,2,\cdots,n-1$; then
\[
R_{k}=f_{S}(k)+(-1)^{\gamma}f_{S}(n-k)+2\sum_{j=1}^{\gamma-1}(-1)^{j}f_{S}(s_{j}-k)
\]
}\end{thm}
\begin{proof}
As in the proof of Theorem \ref{thm:main_aperiodic} let $\delta:=(\delta_{1},\delta_{2},\cdots,\delta_{n+1})$
be the sequence defined by $\delta_{i}:=a_{i}-a_{i-1}$ for all $i=1,2,\cdots,n+1$.
Then $\delta_{1}=a_{1}$, $\delta_{n+1}=(-1)^{\gamma}a_{1}$ and for
$i=1,2,\cdots,n-1$ we have $2f_{S}(i)=a_{1}\delta_{i+1}$. Let $1\leq k<n$.
Hence by Theorem \ref{thm:main_aperiodic} and Lemma \ref{lem:delta-Ck}
\begin{eqnarray*}
R_{k} & = & \frac{1}{2}C_{k}(\delta)=\frac{1}{2}\sum_{i=1}^{n+1-k}\delta_{i}\delta_{i+k}\\
 & = & \frac{1}{2}(\delta_{1}\delta_{k+1}+\delta_{n+1-k}\delta_{n+1}+\sum_{i=2}^{n-k}\delta_{i}\delta_{i+k})\\
 & = & f_{S}(k)+(-1)^{\gamma}f_{S}(n-k)+2\sum_{i=1}^{n-k-1}f_{S}(i)f_{S}(i+k).
\end{eqnarray*}
Since $f_{S}(i)=0$ for $i\geq n$, Theorem \ref{thm:Rk_f} follows
directly from the following Lemma. \end{proof}
\begin{lem}
\label{lem:sumfS}\textup{Let $k=1,2,\cdots,n-1$; then}
\begin{eqnarray*}
\sum_{j=1}^{\gamma-1}(-1)^{j}\cdot f_{S}(s_{j}-k) & = & \sum_{j=1}^{\gamma-1}(-1)^{j}\cdot f_{S}(s_{j}+k)\\
 & = & \sum_{i=1}^{n-1}f_{S}(i+k)\cdot f_{S}(i).
\end{eqnarray*}
\end{lem}
\begin{proof}
We have
\[
\begin{split}\sum_{j=1}^{\gamma-1}(-1)^{j}\cdot f_{S}(s_{j}-k) & =\sum_{j=1}^{\gamma-1}f_{S}(s_{j})\cdot f_{S}(s_{j}-k)=\\
\sum_{i=1}^{n-1}f_{S}(i)\cdot f_{S}(i-k) & =\sum_{i=1}^{n-1}f_{S}(i+k)\cdot f_{S}(i)=\\
\sum_{j=1}^{\gamma-1}f_{S}(s_{j}+k)\cdot f_{S}(s_{j}) & =\sum_{j=1}^{\gamma-1}f_{S}(s_{j}+k)\cdot(-1)^{j}.
\end{split}
\]

\end{proof}
$R_{k}$ can be expressed by Theorem \ref{thm:Rk_f} as the sum of
$\gamma+1$ terms. Each of these terms can be easily computed; for
example by means of a pre-calculated array which holds the values
$f_{S}(i)$ for $i=1,2,\cdots n-1.$ As before we may assume that
$\gamma\approx\frac{n}{2}$ and that $n$ is large. But then, on average
(taken over all binary sequences of length $n$ and all $1\leq k<n$)
only about a quarter of these terms are not zero, since on average
we have $|\{1\leq j<\gamma:\: s_{j}-k>0\;\mbox{\mbox{and}}\; f_{S}(s_{j}-k)\neq0\}|\approx\frac{n-k}{4}$.

If only the first and the last part of the binary sequence $a$ are
known, then we will see that a reformulation of Theorem \ref{thm:Rk_f}
is useful. 
\begin{cor}
\label{cor:Rn-k}Let \textup{$k=1,2,\cdots,n-1$; then} 
\[
(-1)^{\gamma}R_{n-k}=f_{S}(k)+f_{T}(k)+2\sum_{j=1}^{\gamma-1}(-1)^{j}f_{T}(k-s_{j}).
\]
\end{cor}
\begin{proof}
The result follows directly from Theorem \ref{thm:Rk_f}, Lemma \ref{lem:sumfS}
and (\ref{eq:f=00003Dg}). 
\end{proof}
Now let $1\leq k<m\leq n$ and let us assume that the first and the
last $m$ elements of $a$ are known. Then all $s_{j}\in S$ with
$s_{j}<m$ and all $t_{j}\in T$ with $t_{j}<m$ can be easily determined;
the same is true for $f_{S}(i)$ and $f_{T}(i)$ for all $i<m$. Hence
Corollary \ref{cor:Rn-k} can be applied in order to calculate $R_{n-k}$.
In particular, this can be applied in a branch-and-bound like exhaustive
search (as for example in \cite{mertens1999exhaustive}) in order
to find binary sequences with low autocorrelations. 

As an example let $r=(5,2,2,1,2,\cdots,5,3,1,4)$ and assume that
only the first 12 and the last 12 elements of $a$ are known; let
us assume further that $\gamma$ is even. Thus we know that $(s_{1},s_{2},s_{3},s_{4})=(5,7,9,10)$,
$(t_{1},t_{2},t_{3})=(4,5,8)$, $s_{5}\geq12$ and $t_{4}\geq12$.
Thus $f_{S}(j)$ has for $j=1,2,\cdots,11$ the following values:
\[
\begin{array}{ccccccccccc}
0 & 0 & 0 & 0 & -1 & 0 & 1 & 0 & -1 & 1 & 0\end{array}
\]
and $f_{T}(j)$ has for $j=1,2,\cdots,11$ the following values:
\[
\begin{array}{ccccccccccc}
0 & 0 & 0 & -1 & 1 & 0 & 0 & -1 & 0 & 0 & 0\end{array}
\]

For $k<12$ let $S_{k}:=\{k-s>0:\: s\in S\}$ and $T_{k}:=S_{k}\cap T$
then

\[
\begin{array}{cclcccc}
S_{11} & = & \{6,4,2,1\} & \;\; & T_{11} & = & \{4\}\\
S_{10} & = & \{5,3,1\} &  & T_{10} & = & \{5\}\\
S_{9} & = & \{4,2\} &  & T_{9} & = & \{4\}\\
S_{8} & = & \{3,1\}\\
S_{7} & = & \{2\}\\
S_{6} & = & \{1\}
\end{array}
\]
$S_{k}=\textrm{Ø}$ for $k<6$ and $T_{k}=\textrm{Ø}$ for $k<9$.
Thus for $g(j):=f_{S}(j)+f_{T}(j)$ we have $g(4)=g(8)=g(9)=-1$ and
$g(7)=g(10)=1$; all other values of $g(j)$ for $j<12$ are zero.
Note that for $1\leq j<\gamma$ we have $f_{T}(k-s_{j})\neq0$ if
and only if $k-s_{j}\in T_{k}$; hence for $R_{n-k}^{(in)}:=2\sum_{j=1}^{\gamma-1}(-1)^{j}f_{T}(k-s_{j})$
it follows that $R_{n-i}^{(in)}=0$ for $i=1,2,\cdots,8$ and that
$R_{n-9}^{(in)}=-2f_{T}(4)=2$, $R_{n-10}^{(in)}=-2f_{T}(5)=-2$ and
$R_{n-11}^{(in)}=2f_{T}(4)=-2$. Therefore, $R_{n-k}=g(k)+R_{n-k}^{(in)}$
has by Corollary \ref{cor:Rn-k} for $k=1,2,\cdots,11$ the following
values:
\[
\begin{array}{ccccccccccc}
0 & 0 & 0 & -1 & 0 & 0 & 1 & -1 & 1 & -1 & -2.\end{array}
\]

\subsection{Skew-Symmetric Sequences and Run Length Encoding}

An odd length binary sequence $a$ of length $2m-1$ is called \emph{skew-symmetric}
if for $i=1,2,\cdots,m-1$ 
\begin{equation}
a_{m-i}=(-1)^{i}a_{m+i}.\label{eq:skew}
\end{equation}

Skew-symmetric sequences are of particular interest in different areas.
For example, consider the \emph{merit factor} $F(a)$ which is defined
by 
\begin{equation}
F(a):=\frac{n^{2}}{2\sum_{k=1}^{n-1}C_{k}^{2}(a)}.\label{eq:meritFactor}
\end{equation}
In many applications it is of interest to collectively minimize the
absolute values of the autocorrelations; the merit factor can be used
as a possible measure, see \cite{borwein2008merit,jedwab2005surveyspringerlink,jungnickel1999perfect}
for a survey on this topic. Let $F_{n}$ be the highest merit factor
possible for all binary sequences of length $n.$ We say that a binary
sequence $a$ has an \emph{optimal merit factor} if $F(a)=F_{n}$
(where as always $n$ denotes the length of $a$). Many of the known
odd length binary sequences with an optimal merit factor are skew-symmetric
and it is even conjectured in \cite{Golay1055653} that a restriction
to skew-symmetric sequences does not change the asymptotic behavior
of $F_{n}$. Furthermore, all odd length Barker sequences are skew-symmetric;
a \emph{Barker sequence }of length $n$ is a binary sequence with
|$C_{k}(a)|\leq1$ for all $k=1,2,\cdots,n-1$. Since for a binary
sequence $a$ we have $C_{k}(a)+C_{n-k}(a)\equiv4\; mod\; n$ for
$k=1,2,\cdots,n-1$, a Barker sequence has an optimal merit factor.

For a skew-symmetric sequence $a$ it is not difficult to see that
$C_{k}(a)=0$ if $k$ is odd. By (\ref{eq:C1}) it follows that $\gamma=\frac{n+1}{2}=m$.
Furthermore, by Theorem \ref{thm:main_aperiodic} we have $C_{k+1}(a)-2C_{k}(a)+C_{k-1}(a)=-2R_{k}\;\;\mbox{for }k=1,2,\cdots,n-1$.
Therefore, $R_{k}(a)=C_{k}(a)$ if $k>0$ is even and $R_{k}(a)=-\frac{1}{2}(C_{k-1}(a)+C_{k+1}(a))$
if $k$ is odd. In particular, for an odd length Barker sequence $a$
we have $R_{k}(a)=(-1)^{k+\gamma+1}$ for $k=2,3,\cdots,n-1$ and
$R_{1}(a)=-\gamma$ if $\gamma$ is odd and $R_{1}(a)=1-\gamma$ if
$\gamma$ is even. 

Next we want to describe skew-symmetric sequences in terms of their
run length encoding. The run length encoding $r$ of $a$ is called
\emph{skew-symmetric} if $a$ itself is skew-symmetric. Note that
this is well defined. We call the run length encoding $r$ \emph{balanced}
if
\begin{equation}
S\cup T=\{1,2,\cdots,n-1\}\;\mbox{and }\; S\cap T=\varnothing.
\end{equation}
 Note that if $r$ is balanced, then $n=2\gamma-1.$ Furthermore,
$r$ is balanced if and only if we have for each $k=1,2,\cdots,n-1$
either $k\in S$ or $n-k\in S$. Hence if $r$ is balanced, then by
(\ref{eq:s+t}) $f_{S}(k)\neq0\Leftrightarrow f_{T}(k)=0\Leftrightarrow f_{S}(n-k)=0$
for $k=1,2,\cdots,n-1$.

We call a run length encoding $r=(r_{1},r_{2},\cdots,r_{\gamma})$
\emph{reducible} if either $r_{1}=1$ and $r_{\gamma}>1$ or $r_{1}>1$
and $r_{\gamma}=1$. If $r$ is reducible, then the \emph{reduced
run length} encoding $\hat{r}$ is given by $\hat{r}:=(r_{2},r_{3},\cdots,r_{\gamma}-1)$
if $r_{1}=1$ and by $\hat{r}:=(r_{1}-1,r_{2},\cdots,r_{\gamma-1})$
if $r_{1}>1$. Given $r$ we will in the following always denote the
so defined reduced run length encoding by $\hat{r}$. Let $\mathsf{\mathcal{L}}$
denote the set of all run length encodings of binary sequences of
length $n\geq1$. Furthermore, let $\mathsf{\mathcal{L}}_{s}$ be
the set of all $r\in\mathcal{L}$ which are skew-symmetric and let
$\mathsf{\mathcal{L}}_{b}$ be the set of all $r\in\mathcal{L}$ which
are balanced.Now let $r$ be a run length encoding with $\gamma>1$;
then it quite easily follows that 
\begin{eqnarray}
r\in\mathcal{L}_{s} & \Leftrightarrow & r\:\mbox{is reducible and \ensuremath{\hat{r}\in\mathcal{L}_{s}}}\label{eq:reduce_skew}\\
r\in\mathcal{L}_{b} & \Leftrightarrow & r\:\mbox{is reducible and \ensuremath{\hat{r}\in\mathcal{L}_{b}}.}\label{eq:reduce_balance}
\end{eqnarray}

We will show that $\mathsf{\mathcal{L}}_{s}=\mathsf{\mathcal{L}}_{b}$,
i.e.\ that a binary sequence is skew-symmetric if and only if its
run length encoding is balanced. In particular, we will inductively
construct a set $\mathsf{\mathcal{I}}$ such that $\mathsf{\mathcal{L}}_{s}=\mathcal{I}$
and $\mathsf{\mathcal{L}}_{b}=\mathcal{I}$. Let $\mathcal{I}^{(1)}:=\{(1)\}$,
i.e.\ $\mathcal{I}^{(1)}$ is the one-element set containing the
run length encoding of the two binary sequences of length 1. For $k\geq1$
let $\mathcal{I}^{(k+1)}$ be the set of all reducible $r\in\mathcal{L}$
with $\hat{r}\in\mathcal{\mathcal{I}}^{(k)}$ and finally let $\mathcal{I}:=\underset{k\geq1}{\bigcup}\mathcal{I}^{(k)}$.
It is easy to see that the elements of $\mathsf{\mathcal{I}}$ can
be interpreted as the nodes of an infinite binary tree: $r=(1)$ is
the root, and if $r\in\mathcal{\mathcal{I}}$ with $\gamma>1$, then
$\hat{r}$ is its parent. Moreover, $\mathcal{I}^{(k)}$ presents
the set of nodes having depth $k-1$. If $r\in\mathcal{\mathcal{I}}$
then it easily follows by induction that $r\in\mathcal{\mathcal{I}}^{(\gamma)}$.
Furthermore, we have for $\gamma>1$ that $r\in\mathcal{\mathcal{I}}$
if and only if $\hat{r}\in\mathcal{\mathcal{I}}$. 

The next proposition shows that the skew-symmetric sequences are exactly
the binary sequences which have a balanced run length encoding.
\begin{prop}
\label{prop:skew_balanced}$\mathsf{\mathcal{L}}_{s}=\mathsf{\mathcal{L}}_{b}.$ \end{prop}
\begin{proof}
First we will show that $\mathsf{\mathcal{L}}_{s}\subseteq\mathcal{\mathcal{I}}$.
Assume that this is not the case and choose $r\in\mathcal{L}_{s}$
with minimal $\gamma$ such that $r\notin\mathcal{\mathcal{I}}$.
Then $\gamma>1$ and by (\ref{eq:reduce_skew}) $r$ is reducible
and $\hat{r}\in\mathcal{L}_{s}$. Since $r\notin\mathcal{\mathcal{I}}$
we also have $\hat{r}\notin\mathcal{\mathcal{I}}$ which contradicts
the minimality of $\gamma$. Next we want to show that $\mathcal{\mathcal{I\subseteq}}\mathsf{\mathcal{L}}_{s}$.
Similar as before, assume that this is not the case and choose $r\in\mathcal{\mathcal{I}}$
with minimal $\gamma$ such that $r\notin\mathsf{\mathcal{L}}_{s}$.
Then $\gamma>1$, $r$ is reducible and $\hat{r}\in\mathcal{\mathcal{I}}$.
By (\ref{eq:reduce_skew}) $\hat{r}$ is not skew-symmetric which
contradicts the minimality of $\gamma$. Therefore, we have $\mathsf{\mathcal{L}}_{s}=\mathcal{I}$.
Using (\ref{eq:reduce_balance}) similar arguments show that also
$\mathsf{\mathcal{L}}_{s}=\mathcal{I}$ which completes the proof.
\end{proof}

\section{The Periodic Case}

The situation in the periodic case is quite similar to the aperiodic
one. In order to formulate and prove the relationship between the
periodic autocorrelations of binary sequences and their run structure
we have to adapt some of the previous definitions to the periodic
case. In the following we will always assume that $a$ is a binary
sequence of length $n$ which is not constant, i.e.\ $a_{i}\neq a_{j}$
for some $\ensuremath{1\leq i<j\leq n}$. This is no loss of generality
as we will see at the end of this section. We could prove the following
results in a very similar way as we have done in the aperiodic case.
However, by using (\ref{eq:Ckperiodic}) we can directly transfer
the aperiodic results of the previous section to the periodic case. 

In the following we will in addition always assume that $a_{1}\neq a_{n}$.
Again this is no loss of generality since the periodic autocorrelations
are shift-invariant, i.e.\ for $b:=(a_{n},a_{1},a_{2},\cdots,a_{n-1})$
we have $\tilde{C}_{k}(b)=\tilde{C}_{k}(a)$ for all $k=0,1,\cdots,n$.
The main reason for the additional assumption $a_{1}\neq a_{n}$ is
that we can in this case adopt the definition of the the run length
encoding $r=(r_{1},r_{2},\cdots,r_{\gamma})$ of $a$ without any
modification. Note that $\gamma$ is always even if $a_{1}\neq a_{n}$.

\subsection{Preliminaries for the Periodic Case }

As in the periodic case substrings of the run length encoding $r=(r_{1},r_{2},\cdots,r_{\gamma})$
of $a$ play an important role. For $1\leq i\leq\gamma$ we will in
the following always put $r_{\gamma+i}:=r_{i}$. Similar to the concept
of half-open interval on a circle we have to adapt the definition
of a substring to the periodic case. In order to distinguish this
new definition from the previous one we will refer to \emph{p-substrings}.
As we will see, p-substrings $r(i,j)$ of $r$ will be only defined
if and only if $1\leq i,j\leq\gamma$ and $i\neq j$.

For $1\leq i<j\leq\gamma$ the p-substring $r(i,j)$ is the same as
in the periodic case; thus representing $(r_{i},r_{i+1},\cdots,r_{j-1})$.
For $1\leq j<i\leq\gamma$ the p-substring $r(i,j)$ will represent
the non-empty contiguous part $(r_{i},r_{i+1},\cdots,r_{\gamma+j-1})$
of the periodic extension of $r$. Unless otherwise stated we will
in the following always assume that $1\leq i,\; j\leq\gamma$ and
$i\neq j$. As before $|r(i,j)|$ denotes the length of the p-substring,
hence 
\[
|r(i,j)|=\begin{cases}
j-i\; & \mbox{ if }1\leq i<j\leq\gamma\\
\gamma+j-i\; & \mbox{ if }1\leq j<i\leq\gamma.
\end{cases}
\]

In particular we have $1\leq|r(i,j)|<\gamma$. Note that the p-substrings
$r(i,j)$ and $r(j,i)$ are complementary: they have no element in
common and their concatenation represents all elements of $r$; in
particular we have
\begin{equation}
|r(i,j)|+|r(j,i)|=\gamma.\label{eq:r_ij}
\end{equation}

\subsection{The Main Result for the Periodic Case }

The \emph{periodic} \emph{run structure} $\mathsf{\mathcal{\tilde{R}}}$
of $a$ will be defined as the set of all p-substrings of $r$. Similar
to the periodic case we will analyze the connection between the periodic
autocorrelations $\tilde{C}_{k}(a)$ of a binary sequence $a$ of
length $n$ and their periodic run structure given by the run length
encoding $r=(r_{1},r_{2},\cdots,r_{\gamma})$ of $a$. The case $k=1$
is easy; similar to the aperiodic case we have
\begin{equation}
\tilde{C}_{1}(a)=n-2\gamma\label{eq:C1_per}
\end{equation}
since for each $j=1,2,\cdots,\gamma$ the $j$th run of $a$ contributes
$r_{j}-2$ to the sum $\tilde{C}_{1}(a)$. Thus $\tilde{C}_{1}(a)=\sum_{j=1}^{\gamma}(r_{j}-2)=n-2\gamma$. 

The \emph{sum} $S(r(i,j))$ for a p-substring $r(i,j)$ will be defined
as
\[
S(r(i,j))=\begin{cases}
\sum_{m=i}^{j-1}r_{m}\; & \mbox{ if }1\leq i<j\leq\gamma\\
\sum_{m=i}^{\gamma+j-1}r_{m}\; & \mbox{ if }1\leq j<i\leq\gamma.
\end{cases}
\]

Let $1\leq k<n$; we will denote by $\mathsf{\mathcal{\tilde{R}}}_{k}$
the set of all p-substrings $r(i,j)$ of $r$ with $S(r(i,j))=k$
(where as always $1\leq i,j\leq\gamma$ and $i\neq j$) and put
\begin{equation}
\tilde{R}_{k}:=\sum_{u\in\mathcal{\mathit{\mathcal{\tilde{R}}}}_{k}}(-1)^{|u|}.\label{eq:Rk_per}
\end{equation}
 Similar as in the aperiodic case a p-substring of $r$ corresponds
in a one-to-one way to consecutive runs of $a$. Note however, that
\emph{consecutive} has in the periodic case a slightly different meaning
due to the periodic extension of a. Similar to the aperiodic case
$\tilde{R}_{k}$ can be interpreted as the sum of weights of those
consecutive runs of $a$ which (total) length equals $k$; the weight
is in the periodic case either 1 or -1 depending on whether the number
of runs is even or odd. In particular, $-\tilde{R}_{1}$ equals the
number of runs of length 1.

Since $S(r(i,j))+S(r(j,i))=n$, we have $r(i,j)\in\tilde{R}_{k}$
if and only if $r(j,i)\in\tilde{R}_{n-k}$. In particular, this shows
that $|\mathsf{\mathcal{\tilde{R}}}_{k}|=|\mathsf{\mathcal{\tilde{R}}}_{n-k}$|.
Moreover, since $\gamma$ is even, it follows from (\ref{eq:r_ij})
that 
\begin{equation}
(-1)^{|r(i,j)|}=(-1)^{|r(j,i)|}\label{eq:r(j,i)}
\end{equation}
and hence $\tilde{R}_{k}=\tilde{R}_{n-k}$. 

Similar to the periodic case we put $\tilde{R}(a):=(\tilde{R}_{1},\tilde{R}_{2},\cdots,\tilde{R}_{n-1})$.
Let us compute $\tilde{R}(a)$ for the following two sequences of
the previous examples. For $a=(++++++-------)$ we have $n=13$, $\gamma=2$
and $r=(6,7)$. Furthermore, we have $\mathsf{\mathcal{\tilde{R}}}_{k}=\textrm{Ø}$
for $1\leq k\leq12$ unless $k=6$ or $k=7$; in these cases we have
$\mathsf{\mathcal{\tilde{R}}}_{6}=\left\{ r(1,2)\right\} $ and $\mathsf{\mathcal{\tilde{R}}}_{7}=\left\{ r(2,1)\right\} $.
Hence $\tilde{R}(a)=(0,0,0,0,0,-1,-1,0,0,0,0,0)$ by (\ref{eq:Rk_per}).
For $a=(+++------+++---)$ we have $n=15$, $\gamma=4$ and $r=(3,6,3,3)$.
In order to calculate for instance $\tilde{R}_{6}$ we have by (\ref{eq:Rk_per})
to consider all consecutive runs of $a$ which have a total length
of 6. In this case these consecutive runs correspond to the p-substrings
$r(2,3)$, $r(3,1)$ and $r(4,2)$; the p-substrings $r(2,3)$ consists
of a single run, whereas the p-substrings $r(3,1)$ and $r(4,2)$
each consists of two consecutive runs. This shows that $\mathsf{\mathcal{\tilde{R}}}_{6}=\left\{ r(2,3),r(3,1),r(4,2)\right\} $
and $\tilde{R}_{6}=1$. Similarly we have $\mathsf{\mathcal{\tilde{R}}}_{3}=\left\{ r(1,2),r(3,4),r(4,1)\right\} $,
$\mathsf{\mathcal{\tilde{R}}}_{9}=\left\{ r(1,3),r(2,4),r(3,2)\right\} $,
$\mathsf{\mathcal{\tilde{R}}}_{12}=\left\{ r(1,4),r(2,1),r(4,3)\right\} $
and $\mathsf{\mathcal{\tilde{R}}}_{k}=\textrm{Ø}$ if $k$ is not
a multiple of 3; hence $\tilde{R}(a)=(0,0,-3,0,0,1,0,0,1,0,0,-3,0,0)$.
\begin{lem}
\label{lem:Rk}$2\tilde{R}_{k}=R_{k}+R_{n-k}\;\;\mbox{for }k=1,2,\cdots,n-1$.\end{lem}
\begin{proof}
Let $1\leq k<n$. We can write $\mathsf{\mathcal{R}}_{k}$ as the
(disjoint) union $\mathsf{\mathcal{R}}_{k}=\mathsf{\mathcal{R}}_{k}^{(in)}\cup\mathsf{\mathcal{R}}_{k}^{(out)}$
where $\mathsf{\mathcal{R}}_{k}^{(in)}$denotes the set of all inner
substrings in $\mathsf{\mathcal{R}}_{k}$ and $\mathsf{\mathcal{R}}_{k}^{(out)}$denotes
the set of all outer substrings in $\mathsf{\mathcal{R}}_{k}$. Similarly
we can write $\mathsf{\mathcal{\tilde{R}}}_{k}$ as the (disjoint)
union
\[
\mathsf{\mathcal{\tilde{R}}}_{k}=\mathsf{\mathcal{\tilde{R}}}_{k}^{(in)}\cup\mathsf{\mathcal{\tilde{R}}}_{k}^{(out)}\cup\mathsf{\mathcal{\tilde{R}}}_{k}^{(amid)}
\]
 where $\mathsf{\mathcal{\tilde{R}}}_{k}^{(in)}:=\{r(i,j)\in\mathsf{\mathcal{\tilde{R}}}_{k}:1<i<j\leq\gamma\}$,
$\mathsf{\mathcal{\tilde{R}}}_{k}^{(out)}:=\{r(i,j)\in\mathsf{\mathcal{\tilde{R}}}_{k}:i=1\mbox{ or }j=1\}$
and $\mathsf{\mathcal{\tilde{R}}}_{k}^{(amid)}:=\{r(i,j)\in\mathsf{\mathcal{\tilde{R}}}_{k}:1<j<i\leq\gamma\}$.

If $u\in\mathsf{\mathcal{\tilde{R}}}_{k}^{(in)}$ then $u$ is an
inner substring; moreover, we have $u\in\mathsf{\mathcal{\tilde{R}}}_{k}^{(in)}$
if and only if $u\in\mathsf{\mathcal{R}}_{k}^{(in)}$ and thus
\[
2\sum_{u\in\mathcal{\mathit{\mathcal{\tilde{R}}}}_{k}^{(in)}}(-1)^{|u|}=\sum_{u\in\mathcal{\mathit{\mathcal{R}}}_{k}^{(in)}}\alpha(u)\cdot(-1)^{|u|}.
\]

If $r(i,j)\in\mathsf{\mathcal{\tilde{R}}}_{k}^{(amid)}$ then $r(j,i)$
is an inner substring; furthermore we have $r(i,j)\in\mathsf{\mathcal{\tilde{R}}}_{k}^{(amid)}$
if and only if $r(j,i)\in\mathsf{\mathcal{R}}_{n-k}^{(in)}$ and thus
by (\ref{eq:r(j,i)}) 
\[
2\sum_{u\in\mathcal{\mathit{\mathcal{\tilde{R}}}}_{k}^{(amid)}}(-1)^{|u|}=\sum_{u\in\mathcal{\mathit{\mathcal{R}}}_{n-k}^{(in)}}\alpha(u)\cdot(-1)^{|u|}.
\]

For $r(i,j)\in\mathsf{\mathcal{\tilde{R}}}_{k}^{(out)}$we consider
the two possible cases $i=1$ and $j=1$ separately:
\[
r(1,j)\in\mathsf{\mathcal{\tilde{R}}}_{k}^{(out)}\Leftrightarrow r(1,j)\in\mathsf{\mathcal{R}}_{k}^{(out)}\Leftrightarrow r(j,\gamma+1)\in\mathsf{\mathcal{R}}_{n-k}^{(out)}
\]

\[
r(i,1)\in\mathsf{\mathcal{\tilde{R}}}_{k}^{(out)}\Leftrightarrow r(i,\gamma+1)\in\mathsf{\mathcal{R}}_{k}^{(out)}\Leftrightarrow r(1,i)\in\mathsf{\mathcal{R}}_{n-k}^{(out)}.
\]

Hence by (\ref{eq:r(j,i)})
\[
\sum_{u\in\mathcal{\mathit{\mathcal{\tilde{R}}}}_{k}^{(out)}}(-1)^{|u|}=\sum_{u\in\mathcal{\mathit{\mathcal{R}}}_{k}^{(out)}}\alpha(u)\cdot(-1)^{|u|}=\sum_{u\in\mathcal{\mathit{\mathcal{R}}}_{n-k}^{(out)}}\alpha(u)\cdot(-1)^{|u|}.
\]

Summing up we obtain
\[
2\sum_{u\in\mathcal{\mathit{\mathcal{\tilde{R}}}}_{k}}(-1)^{|u|}=\sum_{u\in\mathcal{\mathit{\mathcal{R}}}_{k}}\alpha(u)\cdot(-1)^{|u|}+\sum_{u\in\mathcal{\mathit{\mathcal{R}}}_{n-k}}\alpha(u)\cdot(-1)^{|u|}.
\]

\end{proof}
Since $\gamma$ is even it easily follows from (\ref{eq:sum_rk})
and Lemma \ref{lem:Rk} that $\sum_{k=1}^{n-1}\tilde{R}_{k}=-\gamma$.
The next theorem (cf.\ \cite{cai2009autocorrelation}) which is the
main result of this section shows that for a binary sequence $a$
the periodic autocorrelations and $\tilde{R}(a)$ are closely related.
\begin{thm}
\label{thm:mainPeriodic}Let $k=1,2,\cdots,n-1$; then
\[
\tilde{C}_{k+1}(a)-2\tilde{C}_{k}(a)+\tilde{C}_{k-1}(a)=-4\tilde{R}_{k}.
\]
\end{thm}
\begin{proof}
Let $1\leq k<n$. By Theorem \ref{thm:main_aperiodic} we have $C_{k+1}(a)-2C_{k}(a)+C_{k-1}(a)=-2R_{k}$
and $C_{(n-k)+1}(a)-2C_{n-k}(a)+C_{(n-k)-1}(a)=-2R_{n-k}$. Hence
the theorem follows by (\ref{eq:Ckperiodic}) and Lemma \ref{lem:Rk}.
\end{proof}
By (\ref{eq:C1_per}) we have $\tilde{C}_{1}(a)=n-2\gamma$ and Theorem
\ref{thm:mainPeriodic} gives us that $\tilde{C}_{2}(a)=2\tilde{C}_{1}(a)-n-4\tilde{R}_{1}$
and therefore $\tilde{C}_{2}(a)=n-4\gamma-4\tilde{R}_{1}$. 

Let us denote by $\tilde{C}(a):=(\tilde{C}_{0}(a),\tilde{C}_{1}(a),\cdots,\tilde{C}_{n}(a))$
the \emph{periodic autocorrelation vector} of $a$.
\begin{cor}
\label{cor:delta_per}$\triangle^{2}(\tilde{C}(a))=-4\tilde{R}(a)\;$
for $n\geq3$.\label{cor:-delta-1}\end{cor}
\begin{proof}
This is just a reformulation of Theorem \ref{thm:mainPeriodic}.\end{proof}
\begin{cor}
\label{cor:-Ck-p}Let $k=0,1,\cdots,n$; then
\[
\tilde{C}_{k}(a)=n-2\gamma k-4\sum_{j=1}^{k-1}(k-j)\tilde{R}_{j}.
\]
\end{cor}
\begin{proof}
Since $\tilde{C}_{0}(a)=n$ and by (\ref{eq:C1_per}) $\tilde{C}_{1}(a)=n-2\gamma$
the statement is true for $k=0,1$. For $k\geq2$ it follows directly
from $\triangle^{2}(\tilde{C}(a))=2\tilde{R}(a)$ and Lemma \ref{lem:delta}. 
\end{proof}
Finally, as in the aperiodic case a complementary formulation for
$\tilde{C}_{k}(a)$ in Corollary \ref{cor:-Ck-p} can be easily derived;
since $\tilde{C}_{n}(a)=n$ and $\tilde{C}_{n-1}(a)=\tilde{C}_{1}(a)=n-2\gamma$
equation (\ref{eq:bkplus1}) gives us that for $k=0,1,\cdots,n$ 
\[
\tilde{C}_{k}(a)=n-2\gamma(n-k)-4\sum_{j=k+1}^{n-1}(j-k)\tilde{R}_{j}.
\]

\begin{rem*}
In this section we have for technical reasons assumed that the binary
sequence $a$ is not constant. This however is no loss of generality
since Theorem \ref{thm:mainPeriodic} and the following results are
also true if $a$ is a constant binary sequence. This follows immediately
since for a constant binary sequence $a$ of length $n$ we have $\tilde{C}_{k}(a)=n$
and also $\mathsf{\mathcal{\tilde{R}}}_{k}=\textrm{Ø}$ and thus
$\tilde{R}_{k}=0$ for $k=1,2,\cdots,n-1$. 
\end{rem*}
\bibliographystyle{ieeetr}

\bigskip{}

\end{document}